\documentclass[conference,letterpaper,romanappendices]{ieeeconf}
\let\proof\relax   

\usepackage{amsthm,xpatch}
\usepackage{amsmath,amsfonts}
\usepackage[noadjust]{cite}
\usepackage{amssymb}
\usepackage{dsfont}
\usepackage{graphicx, subfigure}
\usepackage{color}
\usepackage{breqn}
\usepackage{mathtools}
\usepackage{bbm}
\usepackage{latexsym}
\usepackage[ruled, linesnumbered]{algorithm2e}
\usepackage{accents}
\usepackage{enumerate}
\usepackage{tikz}

\newtheorem{lemma}{Lemma}
\newtheorem{theorem}{Theorem}

\newtheorem{example}{Example}

\newcommand*{\transpose}{%
  {\mathpalette\@transpose{}}%
}

\IEEEoverridecommandlockouts

\begin{document}

\newcommand{\SB}[3]{
\sum_{#2 \in #1}\biggl|\overline{X}_{#2}\biggr| #3
\biggl|\bigcap_{#2 \notin #1}\overline{X}_{#2}\biggr|
}

\newcommand{\Mod}[1]{\ (\textup{mod}\ #1)}

\newcommand{\overbar}[1]{\mkern 0mu\overline{\mkern-0mu#1\mkern-8.5mu}\mkern 6mu}

\makeatletter
\newcommand*\nss[3]{%
  \begingroup
  \setbox0\hbox{$\m@th\scriptstyle\cramped{#2}$}%
  \setbox2\hbox{$\m@th\scriptstyle#3$}%
  \dimen@=\fontdimen8\textfont3
  \multiply\dimen@ by 4             
  \advance \dimen@ by \ht0
  \advance \dimen@ by -\fontdimen17\textfont2
  \@tempdima=\fontdimen5\textfont2  
  \multiply\@tempdima by 4
  \divide  \@tempdima by 5          
  \ifdim\dimen@<\@tempdima
    \ht0=0pt                        
    \@tempdima=\fontdimen5\textfont2
    \divide\@tempdima by 4          
    \advance \dimen@ by -\@tempdima 
    \ifdim\dimen@>0pt
      \@tempdima=\dp2
      \advance\@tempdima by \dimen@
      \dp2=\@tempdima
    \fi
  \fi
  #1_{\box0}^{\box2}%
  \endgroup
  }
\makeatother

\makeatletter
\renewenvironment{proof}[1][\proofname]{\par
  \pushQED{\qed}%
  \normalfont \topsep6\p@\@plus6\p@\relax
  \trivlist
  \item[\hskip\labelsep
        \itshape
    #1\@addpunct{:}]\ignorespaces
}{%
  \popQED\endtrivlist\@endpefalse
}
\makeatother

\makeatletter
\newsavebox\myboxA
\newsavebox\myboxB
\newlength\mylenA

\newcommand*\xoverline[2][0.75]{%
    \sbox{\myboxA}{$\m@th#2$}%
    \setbox\myboxB\null
    \ht\myboxB=\ht\myboxA%
    \dp\myboxB=\dp\myboxA%
    \wd\myboxB=#1\wd\myboxA
    \sbox\myboxB{$\m@th\overline{\copy\myboxB}$}
    \setlength\mylenA{\the\wd\myboxA}
    \addtolength\mylenA{-\the\wd\myboxB}%
    \ifdim\wd\myboxB<\wd\myboxA%
       \rlap{\hskip 0.5\mylenA\usebox\myboxB}{\usebox\myboxA}%
    \else
        \hskip -0.5\mylenA\rlap{\usebox\myboxA}{\hskip 0.5\mylenA\usebox\myboxB}%
    \fi}
\makeatother

\xpatchcmd{\proof}{\hskip\labelsep}{\hskip3.75\labelsep}{}{}

\pagestyle{plain}

\title{\fontsize{21}{28}\selectfont On the Capacity of Single-Server Multi-Message \\ Private Information Retrieval with Side Information}


\author{Anoosheh Heidarzadeh, Brenden Garcia, Swanand Kadhe, Salim El Rouayheb, and Alex Sprintson\thanks{A.~Heidarzadeh, B.~Garcia, and A.~Sprintson are with the Department of Electrical and Computer Engineering, Texas A\&M University, College Station, TX 77843 USA (E-mail: \{anoosheh, brendengarcia, spalex\}@tamu.edu).}\thanks{S.~Kadhe is with the Department of Electrical Engineering and Computer Sciences, University of California, Berkeley, CA 94720 USA (E-mail: swanand.kadhe@berkeley.edu).}\thanks{S.~El Rouayheb is with the Department of Electrical and Computer Engineering, Rutgers University, Piscataway, NJ 08854 USA (E-mail: sye8@soe.rutgers.edu).}\thanks{The  work of the fourth author was supported in part by NSF Grant CCF 1817635 and the Google Faculty Research Award. }}


\maketitle 

\thispagestyle{plain}

\begin{abstract} 
We study Private Information Retrieval with Side Information (PIR-SI) in the single-server multi-message setting. In this setting, a user wants to download $D$ messages from a database of $K\geq D$ messages, stored 
on a single server, without revealing any information about the identities of the demanded messages to the server. The goal of the user is to achieve information-theoretic privacy by leveraging the side information about the database. The side information consists of a random subset of $M$ messages in the database which could have been obtained in advance from other users or from previous interactions with the server. The identities of the messages forming the side information are initially unknown to the server. Our goal is to characterize the capacity of this setting, i.e., the maximum achievable download rate.  

In our previous work, we  have established the PIR-SI capacity for the special case in which the user wants a single message, i.e., $D=1$ and showed that the capacity can be achieved through the Partition and Code (PC) scheme. In this paper, we focus  on the case when the user wants multiple messages, i.e., $D>1$. Our first result is that if the user wants more messages than what they have as side information, i.e., $D>M$, then the capacity is $\frac{D}{K-M}$, and it can be achieved using a scheme based on the Generalized Reed-Solomon (GRS) codes. In this case, the user must learn all the messages in the database in order to obtain the desired messages. Our second result shows that this may not be necessary when $D\leq M$, and the capacity in this case can be higher. We present a lower bound on the capacity based on an achievability scheme which we call Generalized Partition and Code (GPC).
\end{abstract}

%



\section{Introduction} 
\label{sec:intro}

In the 
{Private Information Retrieval (PIR) problem}, a user wants to privately download a message belonging to a database with copies stored on a single or multiple remote servers \cite{Chor:PIR1995}. In this paper, we focus on the single-server PIR problem. While the multi-server PIR problem has received extensive attention in the literature (see, e.g., \cite{beimel2002breaking,gasarch2004survey, yekhanin2010private,Sun2017,JafarPIR3new}), 
the model and the information-theoretic privacy guarantees there still hinge on the assumption of restricted collusion among the servers. Not only may this assumption be hard to verify and enforce, but it may also be infeasible in some practical systems; for instance when the servers are owned and operated by same entity. With single-server PIR, the no-collusion assumption is no longer {relevant}.

We focus on information-theoretic PIR~\cite{Sun2017}, which guarantees unconditional privacy irrespective of the computational power of the servers. The drawback of this strong requirement is that in the single-server case, the user has to download the whole database to hide which message they are interested in~\cite{Chor:PIR1995}. This result has led to the aforementioned large body of work on multi-sever (information-theoretic) PIR with a constraint on collusion. 
Aside from information-theoretic PIR is computational PIR which relies mainly on certain limits on the computational power of the server(s), 
assuming the hardness of certain mathematical problems, such as factoring large numbers. 
It was shown in \cite{kushilevitz1997replication} that single-server computational PIR does not require downloading the whole database. Since then, single-server PIR has been almost exclusively studied under the computational privacy model, see e.g.,~\cite{chor1997private,cPIRPoly,gentry2005single}. However, computational PIR schemes typically require homomorphic encryption, and suffer from prohibitive computational cost \cite{sion2007computational}. 

Recently, in~\cite{Kadhe2017}, the authors initiated the study of single-server PIR with Side Information (PIR-SI), wherein the user knows a random subset of messages that is unknown to the server. This side information could have been obtained from other trusted users or from previous interactions with the server.
It was shown that the side information enables the user to achieve  information-theoretic privacy without having to download the entire database. The savings in the download cost depend on whether the user wants to protect the privacy of both the requested message and the messages in the side information, or only the privacy of the requested message. 

\subsection{Main Contributions of This Work}
\label{sec:contributions}


In this work, we consider the single-server multi-message PIR-SI problem in which a user wishes to download $D$ messages from a database of $K$ messages, stored on a single remote server. The user has a random subset of $M$ messages, and the identities of these messages are unknown to the server. We focus on the case in which the user wishes to protect only the identities of the requested messages, and restrict our attention to {PIR-SI} schemes that achieve information-theoretic privacy. 
Our goal is to characterize the capacity of this setting, i.e., the maximum achievable download rate over all such PIR-SI schemes. 



For the regime where $D > M$, we characterize the capacity of the single-server multi-message PIR-SI problem, 
as \mbox{$D/(K-M)$} (Theorem~\ref{thm:1}). 
The converse is based on a mix of combinatorial and information-theoretic arguments, whereas the achievability is based on Generalized Reed-Solomon (GRS) codes~\cite{Roth:06}.
Interestingly, the number of downloaded symbols in this case (i.e., $K-M$) is the same as that in the single-server single-message PIR-SI problem studied in~\cite{Kadhe2017} when the user wishes to protect the identities of messages in both their demand and their side information. This implies that, in the regime $D>M$, protecting only the identities of the demanded messages is as costly (in terms of number of downloaded symbols) as protecting the identities of messages in both the demand and the side information. 
 

When $D\leq M$, we present a lower bound on the capacity of the single-server multi-message PIR-SI problem (Theorem~\ref{thm:2}). 
The achievability results rely on two schemes: (i) 
a scheme termed Generalized Partition and Code (GPC), which is a generalization of the Partition and Code scheme previously presented in~\cite{Kadhe2017}, and (ii) a scheme based on the GRS codes, which is a simple modification of the Maximum Distance Separable (MDS) Code scheme in~\cite{Kadhe2017}. We summarize our main results in Table~\ref{tbl:summary}. The result for $D=1$ is from~\cite{Kadhe2017}, and is presented here for completeness.

\begin{table}[t!]
\caption{Summary of our main results on the capacity of single server PIR-SI.}
\label{tbl:summary}
\centering
\begin{tabular}{|c|c|c|c|}
\hline
Parameters & $D = 1$ & $2\leq D\leq M$ & $D>M$\\
\hline
Capacity & 
\begin{tabular}{@{}c@{}} 
$\left\lceil\frac{K}{M+1}\right\rceil^{-1}$  \\ (\cite[Theorem~1]{Kadhe2017}) 
\end{tabular} 
& 
\begin{tabular}{@{}c@{}} 
Open \\ (Lower bound \\ in Theorem~\ref{thm:2})
\end{tabular} & 
\begin{tabular}{@{}c@{}}$\frac{D}{K-M}$ \\ (Theorem~\ref{thm:1})\end{tabular}\\
\hline
\begin{tabular}{@{}c@{}}
Achievability \\ 
Schemes \end{tabular}
& \begin{tabular}{@{}c@{}}Partition \\ and Code (PC)\end{tabular} &
\begin{tabular}{@{}c@{}} Generalized PC \\ and GRS Code\end{tabular} &
\begin{tabular}{@{}c@{}} GRS \\ Code\end{tabular}\\
\hline
\end{tabular} 
\end{table}

\subsection{Related Work}
\label{sec:related-work}
The replication-based model where multiple servers store copies of the database has been predominantly studied in the PIR literature, with  breakthrough results in the past few years 
 (e.g., \cite{Sun2017, JafarPIR3new, yekhanin2010private, beimel2001information, beimel2002breaking,gasarch2004survey}). The multi-message PIR problem for multiple non-colluding servers was recently considered in~\cite{Banawan2017} under information-theoretic privacy. There has also been a renewed  interest in  PIR for the case in which the data is  stored on the servers using erasure codes, which result in  better storage overhead compared to the traditional replication techniques~\cite{shah2014one, chan2014private, tajeddine2016private, extended, BU18, fazeli2015pir, blackburn2016pir, freij2016private}. 


Most relevant to our setting are works on the capacity of PIR when some side information is present at the user for single server~\cite{Kadhe2017,HKS2018} and multiple servers~\cite{Tandon2017,Wei2017CacheAidedPI,Wei2017FundamentalLO,KGHERS2017,Chen2017side,Maddah2018}. These works differ mainly in their side information models. 
The scenario in which the user can choose the side information, referred to as cache-aided PIR, is considered in~\cite{Tandon2017,Wei2017CacheAidedPI,Wei2017FundamentalLO} for the case of multiple servers. In~\cite{Tandon2017}, the side information can be any function of the database, and is known to all the servers. In~\cite{Wei2017CacheAidedPI}, the side information is uncoded and not known to the servers, whereas, in~\cite{Wei2017FundamentalLO}, the servers are partially aware of the uncoded side information. 

The single-server single-message scenario where the side information available at the user is in the form of a random subset of messages is considered in~\cite{Kadhe2017}. 
In~\cite{KGHERS2017}, it is shown that this type of side information also helps in multi-server single-message scenario when the user only wants to protect the requested message. For the same type of side information, the multi-server single-message and the multi-server multi-message scenarios are respectively studied in~\cite{Chen2017side} and~\cite{Maddah2018}, when the user wants to protect both the requested message(s) and the side information messages. The single-server single-message case where the side information is a random linear combination of a random subset of messages is considered in~\cite{HKS2018} when the user wants to protect the requested message. 



\section{Problem Formulation}\label{sec:PF}
For a positive integer $i$, we denote $\{1,\dots,i\}$ by $[i]$. With an abuse of notation, we denote a random variable and its realizations by the same uppercase letter whenever clear from the context. With an abuse of notation we denote a random variable and its realization by the same uppercase letter whenever clear from the context; whenever clarification is needed we denote a random variable with a bold symbol, e.g. $\boldsymbol{X}$, and its realization without bold face, e.g., $X$. Also, we denote the (Shannon) entropy of a random variable $X$ by $H(X)$, and the conditional entropy of a random variable $X$ given a random variable $Y$ by $H(X|Y)$. 


Let $\mathbb{F}_{q^m}$ be an extension field of the finite field $\mathbb{F}_q$ for some prime $q\geq 2$ and $m\geq 1$. 
There is a server storing a set $X$ of $K\geq 1$ messages, $X \triangleq \{X_1,\dots,X_K\}$, where each $X_i$ is chosen independently and uniformly from $\mathbb{F}_{q^m}$, i.e., $H(X_1)=H(X_2)=\dots=H(X_K)=L$ and $H(X_1,X_2,\dots,X_K)=KL$ where $L \triangleq m \log_2 q$.
For any $T\subseteq [K]$, denote $\{X_i\}_{i\in T}$ by $X_T$. 
There is a user that wishes to retrieve $D$ ($1\leq D\leq K$) messages $X_W$ from the server for some $W \subseteq [K]$, $|W|=D$. Additionally the user knows $M$ ($0 \leq M \leq K-D$) messages $X_S$ for some $S \subseteq [K]\setminus W$, $|S|=M$. 
(Note that $W$ and $S$ are disjoint sets.) We refer to $W$ as \emph{demand index set}, $X_W$ as \emph{demand}, and $D$ as \emph{demand size}, and refer to $S$ as \emph{side information index set}, $X_S$ as \emph{side information}, and $M$ as \emph{side information size}. 


In this work, we assume that $\boldsymbol{W}$ is distributed uniformly over all subsets of $[K]$ of size $D$,
\[\mathbb{P}(\boldsymbol{W}=W) = \frac{1}{\binom{K}{D}}, \quad W\subset [K], |W|=D,\] and $\boldsymbol{S}$ is uniformly distributed over $[K]\setminus W$, 
\begin{equation*}
\mathbb{P}(\boldsymbol{S}=S|\boldsymbol{W}=W) = 
\left\{\begin{array}{ll}
\frac{1}{\binom{K-D}{M}}, & S\subseteq [K]\setminus W,\\
0, & \text{otherwise}.
\end{array}\right.	
\end{equation*} 


The server does not know the realizations $S$ and $W$ a priori, whereas the size of $W$ (i.e.. $D$), the size of $S$ (i.e., $M$), and the distribution of $\boldsymbol{W}$ as well as the conditional distribution of $\boldsymbol{S}$ given $\boldsymbol{W}$ are known to the server a priori.

Given $S$, $W$, and $X_S$, to retrieve $X_W$, the user sends a query $Q^{[W,S]}$ to the server, which is a (potentially stochastic) function of $W$, $S$, and $X_S$, and is independent of $X_{[K]\setminus S}$, the messages that are not in the user's side information. 

The query $Q^{[W,S]}$ is required to protect the privacy of the user's demand index set $W$ at the server, i.e., \[\mathbb{P}(\boldsymbol{W}= W^{*}| Q^{[W,S]},X)= \frac{1}{\binom{K}{D}}\quad W^{*} \subset [K], |W^{*}|=D.\] We refer to this condition as the \emph{privacy condition}. 
This condition is also known as \emph{$W$-privacy} \cite{Kadhe2017}. A stronger type of privacy is \emph{$(W,S)$-privacy}, also introduced in~\cite{Kadhe2017}, where both the user's demand and side information index sets must remain private to the server. Although our main focus in this work is on $W$-privacy, we will show that $W$-privacy implies $(W,S)$-privacy in certain cases, depending on $D$ and $M$.

Upon receiving $Q^{[W,S]}$, the server sends the user an answer $A^{[W,S]}$, which is a (deterministic) function of $Q^{[W,S]}$ and $X$, i.e., \[H({A}^{[W,S]}| Q^{[W,S]},X) = 0.\] 

The answer $A^{[W,S]}$ together with $X_S$ must enable the user to retrieve $X_W$, i.e., \[H({X}_W| A^{[W,S]}, Q^{[W,S]}, X_S)=0.\] We refer to this condition as the \emph{recoverability condition}. 


The \emph{single-server multi-message Private Information Retrieval with Side Information (PIR-SI)} problem is to design a query $Q^{[W,S]}$ and an answer $A^{[W,S]}$ for any $W$ and $S$ that satisfy the privacy and recoverability conditions.

We refer to a collection of $Q^{[W,S]}$ and $A^{[W,S]}$ for all $W$ of fixed size $D$ and all $S$ of fixed size $M$ as a \emph{PIR-SI protocol}. 


The \emph{rate} of a PIR-SI protocol is defined as the ratio of the entropy of $D$ messages, i.e., $DL$, to the average entropy of the answer, i.e., $H({A}^{[\boldsymbol{W},\boldsymbol{S}]})=\sum H({A}^{[W,S]})\mathbb{P}(\boldsymbol{W}=W,\boldsymbol{S}=S)$, where the sum is over all $W$ of size $D$ and all $S$ of size $M$. 

The \emph{capacity} of the PIR-SI problem is defined as the supremum of rates over all PIR-SI protocols.

In this work, our goal is to characterize (or establish a lower bound on) the capacity of PIR-SI problem for all $D$ and $M$, and to design a PIR-SI protocol that achieves the capacity (or the derived lower bound on the capacity). 

\section{Main Results}\label{sec:MR}
In this section, we present our main results. Theorem~\ref{thm:1} characterizes the capacity of PIR-SI problem for $D>M$, and Theorem~\ref{thm:2} provides a lower bound on the capacity of PIR-SI problem for $D\leq M$. The proofs of Theorems~\ref{thm:1} and~\ref{thm:2} can be found in Sections~\ref{sec:PIRSI-I} and~\ref{sec:PIRSI-II}, respectively. 

\begin{theorem}\label{thm:1}
The capacity of the PIR-SI problem with $K$ messages, side information size $M$, and demand size $D>M$ is given by $D/(K-M)$.
\end{theorem}

The proof of converse is by upper bounding the rate of any PIR-SI protocol for $D>M$ using a mixture of combinatorial and information-theoretic arguments based on a necessary condition imposed by both the privacy and recoverability conditions 
(see Lemma~\ref{prop:1}). The achievability proof relies on a PIR-SI scheme based on the Generalized Reed-Solomon (GRS) codes~\cite{Roth:06}, termed GRS Code protocol, which is a simple modification of the Maximum Distance Separable (MDS) Code scheme of~\cite{Kadhe2017} and achieves the derived upper bound on the rate (see Section~\ref{sec:PIRSI-I}). 

\textbf{Remark 1.} 
It was shown in~\cite{Chor:PIR1995} that when there is a single server storing $K$ messages, and the user demands a single message ($D=1$), and does not know any of the messages a priori ($M=0$), in order to guarantee the privacy of the demand index set, the user needs to download $K$ units of information, e.g., all the $K$ messages; and thus the capacity is $1/K$. This result matches the result of Theorem~\ref{thm:1} for $D=1$ and $M=0$. In general, the result of Theorem~\ref{thm:1} shows that when the user knows $M$ messages a priori, and demands $D>M$ other messages, the capacity is $D/(K-M)$. That is, the user needs to download $K-M$ units of information that they do not know a priori, e.g., $K-M$ MDS-coded combinations of all the $K$ messages. 


\textbf{Remark 2.}
The result of Theorem~\ref{thm:1} for $D=1$ matches the capacity of single-server single-message PIR-SI problem, studied in~\cite{Kadhe2017}, where $(W,S)$-privacy is required. Moreover, the result of~\cite{Kadhe2017} for single-server single-message PIR-SI problem when $(W,S)$-privacy is required naturally extends to the multi-message setting, i.e., the capacity of single-server multi-message PIR-SI problem when $(W,S)$-privacy is required is $D/(K-M)$. 
By comparing this result with that of Theorem~\ref{thm:1}, it can be seen that for $D>M$, the capacity of single-server multi-message PIR-SI problem when $W$-privacy is required is the same as that of single-server multi-message PIR-SI problem when $(W,S)$-privacy is required. 

\begin{theorem}\label{thm:2}
The capacity of the PIR-SI problem with $K$ messages, side information size $M$, and demand size $D\leq M$ is lower bounded by 
\[\max\left\{\frac{D}{K-M}, D\left(K-\left\lfloor\frac{M}{D}\right\rfloor\left\lfloor\frac{K}{D+\lfloor M/D\rfloor}\right\rfloor\right)^{-1}\right\}\] if $\frac{K-D}{D+\lfloor M/D\rfloor}\leq \left\lfloor\frac{K}{D+\lfloor M/D\rfloor}\right\rfloor$, and \[\max\left\{\frac{D}{K-M}, \left\lceil\frac{K}{D+\lfloor M/D\rfloor}\right\rceil^{-1}\right\}\] if $\frac{K-D}{D+\lfloor M/D\rfloor}> \left\lfloor\frac{K}{D+\lfloor M/D\rfloor}\right\rfloor$.
\end{theorem}

For those cases of $K$, $D$, and $M$ such that the capacity is lower bounded by $D/(K-M)$, the proof relies on the GRS Code protocol (see Section~\ref{sec:PIRSI-I}). In these cases, not only is $W$-privacy achieved but also $(W,S)$-privacy is achieved. For the rest of the cases of $K$, $D$, and $M$, the proof is based on constructing a PIR-SI protocol that achieves the rate ${D/\left(K-\left\lfloor M/D \right\rfloor \lfloor K/(D+\lfloor M/D\rfloor)\rfloor\right)}$ or ${1/\lceil K/(D+\lfloor M/D\rfloor)\rceil}$, depending on $K$, $D$, and $M$. This protocol, referred to as the Generalized Partition and Code (GPC) protocol, is a generalization of the Partition and Code protocol previously introduced in~\cite{Kadhe2017} (see Section~\ref{sec:PIRSI-II}).
\textbf{Remark 3.}
Under a few divisibility assumptions, the results of Theorem~\ref{thm:2} can be further simplified. In particular, when $D$ divides $M$, and $D+\lfloor M/D\rfloor$ divides $K$, the lower bound on the capacity is ${\max\{D/(K-M),(D+M/D)/K\}}$. For a given $D$, if $M$ is sufficiently large, namely, ${M> K-D^2}$, the lower bound is equal to ${D/(K-M)}$, and it can be achieved by the GRS Code protocol (see Section~\ref{sec:PIRSI-I}); otherwise, if ${D\leq M\leq K-D^2}$, the lower bound is equal to ${(D+ M/D)/K}$, and it can be achieved by the GPC protocol (see Section~\ref{sec:PIRSI-II}) .  

\textbf{Remark 4.}
For $D>1$, the tightness of the lower bounds in Theorem~\ref{thm:2} remains unknown in general. However, these lower bounds are tight for ${D=1}$. In this case, by Theorem~\ref{thm:2}, the capacity is lower bounded by ${1/(K-M\lfloor K/(M+1)\rfloor)}$ if ${(K-1)/(M+1)\leq \lfloor K/(M+1)\rfloor}$, 
and is lower bounded by ${1/\lceil K/(M+1)\rceil}$ if ${(K-1)/(M+1)> \lfloor K/(M+1)\rfloor}$. 
It can be shown that both lower bounds are equal to ${1/\lceil K/(M+1)\rceil}$ (see Appendix). 
It was also shown in~\cite{Kadhe2017} that for $D=1$, the capacity is equal to ${1/\lceil K/(M+1)\rceil}$. This shows that the result of  Theorem~\ref{thm:2} is tight for $D=1$.



\section{PIR-SI Problem: The case of $D>M$}\label{sec:PIRSI-I}

\subsection{Proof of Converse for Theorem~\ref{thm:1}}

\begin{lemma}\label{lem:1}
The capacity of the PIR-SI problem for $D>M$ is upper bounded by $D/(K-M)$. 	
\end{lemma}

\begin{proof}[Proof (Sketch)]
Suppose that the user wishes to retrieve $X_W$ for a given $W\subset [K]$, $|W|=D$, and it knows $X_S$ for a given $S\subseteq [K]\setminus W$, $|S|=M$. Note that $D>M$ (by assumption). The user sends to the server a query $Q \triangleq Q^{[W,S]}$, and the server responds to the user by an answer $A \triangleq A^{[W,S]}$. We need to show that $H({A})\geq (K-M)L$. 

The proof is based on several results, formally proved later, and here we present a sketch of the proof to provide more intuition about the subsequent results.  

First, we prove a necessary condition imposed by the privacy and recoverability conditions. Specifically, we show that given $Q$ and $X$ (and $A$), for any candidate demand index set $W^{*}\subset [K]$, $|W^{*}|=D$, the server must be able to find a potential side information index set $S^{*}\subseteq [K]\setminus W^{*}$, ${|S^{*}|=M}$ such that if the user's demand index set was $W^{*}$, then the user could recover $X_{W^{*}}$ from $Q$, $A$ and $X_{S^{*}}$. If not, the server learns that $W^{*}$ is not the user's demand index set, and this violates the privacy condition. This observation, formally stated in Lemma~\ref{prop:1}, is one of the key components in the proof. Based on this observation, we then show that for $D>M$ the user must be able to recover all messages $X_{[K]\setminus S}$ from $A$ and $X_S$, i.e., $H({X}_{[K]\setminus S}|A,Q,X_{S})=0$. This is the main idea of the proof, and will be formally proved in Lemmas~\ref{lem:B1} and~\ref{lem:B2}. The rest of the proof proceeds as follows.

By the chain rule of entropy, it is easy to show that
\begin{equation*}
H(A|Q,X_{S}) = H(X_{[K]\setminus S}|Q,X_{S}).
\end{equation*} Given $Q$ and $X_S$, the user has no knowledge about $X_{[K]\setminus S}$. Thus, $X_{[K]\setminus S}$ is independent of $(Q,X_S)$, i.e., $H(X_{[K]\setminus S}|Q,X_{S}) = H(X_{[K]\setminus S})$. Then,
\begin{equation*}
H(A|Q,X_{S}) = H(X_{[K]\setminus S}) = (K-M)L	
\end{equation*} since $H(X_{[K]\setminus S}) = \sum_{i\in [K]\setminus S} H(X_i) = (K-M)L$ (by the uniformity and independence of $X_1,\dots,X_K$). Since conditioning does not increase the entropy, then 
\begin{equation*}
H(A)\geq H(A|Q,X_{S}) = (K-M)L,	
\end{equation*} as was to be shown.
\end{proof}

In the sequel, we provide the proofs of the main ingredients discussed in the proof sketch of Lemma~\ref{lem:1}. 

The following result gives a necessary condition for privacy and recoverability.

\begin{lemma}\label{prop:1}
For any $W\subset [K]$, $|W|=D$, and ${S\subseteq [K]\setminus W}$, $|S|=M$, and any ${W^{*}\subset [K]}$, $|W^{*}|=D$, there must exist ${S^{*}\subseteq [K]\setminus W^{*}}$, $|S^{*}|=M$ such that \[H(X_{W^{*}}| A^{[W,S]}, Q^{[W,S]}, X_{S^{*}}) = 0.\] 
\end{lemma}


\begin{proof}
The proof is by contradiction. Suppose that there does not exist any $S^{*}$ (and correspondingly, $X_{S^{*}}$) such that $X_{W^{*}}$ is recoverable from $A^{[W,S]}$, $Q^{[W,S]}$, and $X_{S^{*}}$. Then, the server knows that $W^{*}$ cannot be the user's demand index set, and this violates the privacy condition.
\end{proof}


The following lemma shows that given $Q$ and $X$ (and $A$), for $D>M$ the server must be able to find a candidate side information index set $S_{*}\subset [K]$, $|S_{*}|=M$ such that if the user's side information index set was $S_{*}$, then the user could recover all messages $X_{[K]\setminus S_{*}}$ from $Q$, $A$, and $X_{S_{*}}$.

\begin{lemma}\label{lem:B1}
For any $W\subset [K]$, $|W|=D>M$, and ${S\subseteq [K]\setminus W}$, $|S|=M$, there exists $S_{*}\subset [K]$, $|S_{*}|=M$ such that \[H(X_{[K]\setminus S_{*}}|A^{[W,S]},Q^{[W,S]},X_{S_{*}})=0.\] 
\end{lemma}

Before proving Lemma~\ref{lem:B1}, we present a toy example to give more intuition about the result of the lemma. 




\begin{example}\label{ex:I1}
Consider a scenario where ${K=5}$, ${D=2}$, and ${M=1}$. Suppose that the server has the messages ${X_1,\dots,X_5}$, and the user knows one message, say $X_1$, and demands two other messages, say ${X_2,X_3}$. That is, ${W=\{2,3\}}$ and ${S=\{1\}}$. 

By the result of Lemma~\ref{prop:1}, for any two messages $X_{i_1},X_{i_2}$, the server must be able to find another message $X_{i}$, ${i\not\in \{i_1,i_2\}}$, as a potential side information, i.e., $X_{i_1},X_{i_2}$ are recoverable from $Q\triangleq Q^{[\{2,3\},\{1\}]}$, $A\triangleq A^{[\{2,3\},\{1\}]}$, and $X_i$. Define the notation \[\{i\}\rightarrow \{i_1,\dots,i_n\}\] to mean that $X_{i_1},\dots,X_{i_{n}}$ are recoverable from $Q$, $A$, and $X_i$. 

By the recoverability condition, ${\{1\}\rightarrow \{2,3\}}$, i.e., $X_2,X_3$ must be recoverable from $Q$, $A$, and $X_1$. 

Now, consider the candidate demand index set $\{1,4\}$. There are two cases: 
\begin{itemize}
\item[(i)] ${\{5\}\rightarrow \{1,4\}}$. Since we also have $\{1\}\rightarrow \{2,3\}$, then ${\{5\}\rightarrow \{1,2,3,4\}}$.  
\item[(ii)]  ${\{2\}\rightarrow \{1,4\}}$ (or ${\{3\}\rightarrow \{1,4\}}$).
Similarly, since we also have ${\{1\}\rightarrow \{2,3\}}$, then ${\{2\}\rightarrow \{1,3,4\}}$. Now, consider the candidate demand index set $\{4,5\}$. Either ${\{1\}\rightarrow \{4,5\}}$, or ${\{2\}\rightarrow \{4,5\}}$, or ${\{3\}\rightarrow \{4,5\}}$. In all of these cases, ${\{2\}\rightarrow \{1,3,4,5\}}$. 
\end{itemize}

Taking $S_{*}=\{5\}$ or $S_{*}=\{2\}$ in case (i) or case (ii), respectively, it follows that the result of Lemma~\ref{lem:B1} holds for ${K=5}$, ${D=2}$, and ${M=1}$. 
\end{example}

Next, we prove Lemma~\ref{lem:B1} by extending the reasoning in Example~\ref{ex:I1} to arbitrary $K$, $D$, and $M$ such that ${D>M}$.

\begin{proof}[Proof of Lemma~\ref{lem:B1}]
It suffices to show that there exists ${W_{*}\subset [K]}$, ${|W_{*}|=K-M}$ such that all the messages in ${X_{W_{*}}}$ are recoverable from $Q\triangleq Q^{[W,S]}$, $A \triangleq A^{[W,S]}$, and $X_{[K]\setminus W_{*}}$. Taking ${S_{*}=[K]\setminus W_{*}}$, the proof will be complete. 







The idea is to start with a set of $D$ messages that can be recovered from $M$ ($<D$) other messages (along with $Q$ and $A$), and grow this set in size (up to $K-M$) recursively. 







Let ${W_1\subset [K]}$, ${D\leq |W_1|<K-M}$ and ${S_1\subseteq [K]\setminus W_1}$, ${|S_1|=M}$ be such that \[{H(X_{W_1}|A,Q,X_{S_1})=0}.\] (Such $W_1$ and $S_1$ exist; in particular, $W_1=W$ and $S_1=S$ satisfy these requirements (by the recoverability condition).) We will show that there exist ${W_2\subset [K]}$, ${|W_2|>|W_1|}$ and ${S_2\subseteq [K]\setminus W_2}$, ${|S_2|=M}$ such that \[{H(X_{W_2}|A,Q,X_{S_2})=0}.\] That is, starting with any arbitrary set of $|W_1|$ (${<K-M}$) messages $X_{W_1}$ which are recoverable from $M$ other messages $X_{S_1}$, one can always find a set of $|W_2|$ ($>|W_1|$) messages $X_{W_2}$ which can be recovered from a set of $M$ other messages $X_{S_2}$. The proof is as follows.  

Take an arbitrary ${i \in [K] \setminus (W_1 \cup S_1)}$. (Such an index $i$ exists because ${|W_1\cup S_1|<K}$.) Note that ${|S_1\cup \{i\}|=M+1\leq D}$ since $D>M$. Take an arbitrary $W_0\subseteq [K]\setminus (S_1\cup \{i\})$, $|W_0|=D-M-1$. Note that $|W_0|\geq 0$. By the result of Lemma~\ref{prop:1}, there must exist ${S_2\subseteq [K]\setminus (S_1\cup \{i\}\cup W_0)}$, ${|S_2|=M}$ such that \[{H(X_{S_1},X_i,X_{W_0}|A,Q,X_{S_2})=0}.\] Let ${W_2 = (W_1 \setminus S_2) \cup (S_1 \cup \{i\}\cup W_0)}$. It remains to show that $W_2$ and $S_2$ satisfy the following requirements: (i) ${|W_2|>|W_1|}$; (ii) ${S_2\subseteq [K]\setminus W_2}$, and (iii) ${H(X_{W_2}|A,Q,X_{S_2})=0}$. It is easy to see that \[{|W_2|= |W_1|-|S_2|+|S_1|+1+|W_0|-|W_1\cap W_0|},\] and subsequently, ${|W_2|\geq |W_1|+1>|W_1|}$. It is also easy to see that ${S_2\subseteq [K]\setminus W_2}$ because ${S_2\cap (W_1\setminus S_2) = \emptyset}$ and ${S_2\cap (S_1\cup \{i\}\cup W_0)=\emptyset}$. Moreover,
\begin{align*}
& H(X_{W_2}|A,Q,X_{S_2}) \\ & \quad = H(X_{S_1},X_i,X_{W_0}|A,Q,X_{S_2}) \\ 
&\quad\quad +H(X_{W_1 \setminus S_2}|A,Q,X_{S_2},X_{S_1},X_i,X_{W_0})\\ 
&\quad =0.	
\end{align*} Thus, $W_2$ and $S_2$ satisfy the requirements~(i)-(iii). 

By repeating the above arguments recursively, it follows that there exists ${W_{*}\subset [K]}$, ${|W_{*}|= K-M}$ such that ${H(X_{W_{*}}|A,Q,X_{[K]\setminus W_{*}})=0}$, as desired. 
\end{proof}



Note that Lemma~\ref{lem:B1} guarantees that for $D>M$ there exists a ``potential'' side information index set $S_{*}$ such that if the user's side information index set was $S_{*}$, then the user could recover all messages $X_{[K]\setminus S_{*}}$ from $Q$, $A$, and $X_{S_{*}}$. However, it is not obvious that the user's ``actual'' side information index set $S$ is one such set $S_{*}$. In the following lemma, we show that this must be the case for ${D>M}$. 

\begin{lemma}\label{lem:B2}
For any $W\subset [K]$, $|W|=D$ ($>M$), and ${S\subseteq [K]\setminus W}$, $|S|=M$, \[H(X_{[K]\setminus S}|A^{[W,S]},Q^{[W,S]},X_{S})=0.\]
\end{lemma}

The following example explains the result of Lemma~\ref{lem:B2} for the scenario of Example~\ref{ex:I1}.


\begin{example}\label{ex:I2}
Consider the scenario of Example~\ref{ex:I1}. Recall that in case~(i), $\{1\}\rightarrow \{2,3\}$ and $\{5\}\rightarrow\{1,4\}$, and in case~(ii), $\{1\}\rightarrow\{2,3\}$, $\{2\}\rightarrow \{1,4\}$, and $\{1\}\rightarrow \{4,5\}$. 

First, consider the case~(i). Consider the candidate demand index set $\{4,5\}$. Either $\{1\}\rightarrow \{4,5\}$, or $\{2\}\rightarrow \{4,5\}$, or $\{3\}\rightarrow \{4,5\}$. In either case, $\{1\}\rightarrow \{2,3,4,5\}$ since $\{1\}\rightarrow \{2,3\}$. 

Next, consider the case~(ii). Since $\{1\}\rightarrow \{2,3\}$ and $\{1\}\rightarrow \{4,5\}$, then $\{1\}\rightarrow \{2,3,4,5\}$. 

Thus, in both cases (i) and (ii), all the messages $X_2,\dots,X_5$ can be recovered from $Q$, $A$, and $X_1$. Noting that $X_1$ is the user's actual side information, it follows that the result of Lemma~\ref{lem:B2} holds for ${K=5}$, ${D=2}$, and ${M=1}$. 
\end{example}


Lemma~\ref{lem:B2} generalizes the result of Example~\ref{ex:I2} for arbitrary $K$, $D$, and $M$ such that $D>M$, and we prove this lemma by way of contradiction as follows.


\begin{proof}[Proof of Lemma~\ref{lem:B2}]
Fix arbitrary $W$ and $S$. Let $Q \triangleq Q^{[W,S]}$ and $A \triangleq A^{[W,S]}$. Let $\mathcal{S}$ be the collection of all ${S_{*}\subset [K]}$, ${|S_{*}|=M}$ such that ${H(X_{[K]\setminus S_{*}}|A,Q,X_{S_{*}})=0}$. Since $A$, $Q$, and $X$ are available at the server, then $\mathcal{S}$ is known to the server. From the perspective of the server, there are two possibilities: the user can recover all messages in $X$ from $Q$ and $A$, i.e., $H(X_{S_{*}}|A,Q,X_S)=0$ for some $S_{*}\in \mathcal{S}$, or given $Q$ and $A$, the user cannot recover all messages in $X$, i.e., ${H(X_{S_{*}}|A,Q,X_S)> 0}$ for all $S_{*}\in \mathcal{S}$. 

First, suppose that $H(X_{S_{*}}|A,Q,X_S)=0$ for some $S_{*}\in \mathcal{S}$. Then, $S$ must belong to $\mathcal{S}$, i.e., $H(X_{[K]\setminus S}|A,Q,X_{S_{*}})=0$, as was to be shown. Next, suppose that $H(X_{S_{*}}|A,Q,X_S)> 0$ for all $S_{*}\in \mathcal{S}$. That is, $S\not\in \mathcal{S}$. In this case, we will show a contradiction. 
Since $S$ is the user's (actual) side information index set and $S\not\in \mathcal{S}$, from the server's perspective the user's side information index set cannot belong to $\mathcal{S}$. Thus, from the perspective of the server, none of $S_{*}\in \mathcal{S}$ can be the user's side information index set. 
That is, for a given subset of $D$ indices, say $W^{*}$, if the server can only pair $W^{*}$ with some candidate side information index set(s) in $\mathcal{S}$ (given $Q$ and $A$), then $W^{*}$ cannot be the user's (actual) demand index set. Since this violates the privacy condition and the server knows that $Q$ and $A$ satisfy the privacy condition, then the server simply rules out any $S_{*}\in \mathcal{S}$ from the set of candidate side information index sets for any possible demand index set. Repeating the same lines as in the proof of Lemma~\ref{lem:B2} except when ruling out all $S_{*}\in \mathcal{S}$ from the set of candidate side information index sets, it follows that there must exist ${S_0\not\in \mathcal{S}}$, ${S_0\subset [K]}$, ${|S_0|=M}$ such that ${H(X_{[K]\setminus S_0}|A,Q,X_{S_0})=0}$. Since $\mathcal{S}$ contains all ${S_{*}\subset [K]}$, ${|S_{*}|=M}$ such that ${H(X_{[K]\setminus S_{*}}|A,Q,X_{S_{*}})=0}$, then $\mathcal{S}$ must contain $S_0$. i.e., ${S_0\in \mathcal{S}}$. This is a contradiction. 
\end{proof}

\subsection{Proof of Achievability for Theorem~\ref{thm:1}}


In this section, we propose a PIR-SI protocol for arbitrary $K$, $M$, and $D$. We notice that this protocol is applicable to both cases of $D>M$ and $D\leq M$. As will be shown shortly, the rate achieved by this protocol is optimal for $D>M$, whereas for $D\leq M$ the protocol proposed in Section~\ref{sec:PIRSI-II} may achieve a higher rate, depending on $K$, $D$, and $M$. 

Assume that $q\geq K$, and let $\omega_1,\dots,\omega_{K}$ be $K$ distinct elements from $\mathbb{F}_q$. 

\textit{\bf Generalized Reed-Solomon (GRS) Code Protocol:} This protocol consists of three steps as follows: 

\textbf{\it Step 1:} The user constructs $K-M$ sequences $Q_{1},\dots,Q_{K-M}$, each of length $K$, such that $Q_{i}  = \{\omega^{i-1}_1,\dots,\omega^{i-1}_{K}\}$ for $i\in [K-M]$, and sends to the server the query $Q^{[W,S]} = \{Q_{1},\dots,Q_{K-M}\}$. 

Note that for any ${i\in [K-M]}$ and ${j\in [K]}$, the $j$th element in the sequence $Q_i$ can be thought of as the entry $(i,j)$ in a ${(K-M)\times K}$ Vandermonde matrix $V$ with distinct parameters ${\omega_1,\dots,\omega_{K}}$. It is also well-known that the matrix $V$ is full-rank, and any ${(K-M)\times (K-M)}$ sub-matrix of $V$ is full-rank~\cite{Roth:06}. 


\textbf{\it Step 2:} By using $Q_i$, the server computes $A_i = \sum_{j=1}^{K} \omega^{i-1}_{j} X_{j}$ for $i\in [K-M]$, and sends to the user the answer $A^{[W,S]}=\{A_{1},\dots,A_{K-M}\}$. In the language of coding theory, the matrix $V$ can be viewed as the parity-check matrix of a $(K,M)$ (normalized) Generalized Reed-Solomon (GRS) code~\cite{Roth:06}, and accordingly, $A_1,\dots,A_{K-M}$ are the parity check equations of this code (and hence the name of the proposed protocol).

\textbf{\it Step 3:} Upon receiving the answer from the server, the user retrieves $X_j$ for each $j\in [K]\setminus S$ by subtracting off the contribution of side information $\{X_i\}_{i\in S}$ from the $K-M$ equations $A_{1},\dots,A_{K-M}$, and solving the resulting system of $K-M$ linear equations with $K-M$ unknowns.
 
{\textbf{Remark 5.}} It should be noted that the GRS Code protocol is similar to the MDS Code scheme in~\cite{Kadhe2017}. The main benefit of the GRS Code protocol is that it can operate over any field of size $q \geq K$, whereas the MDS Code scheme of~\cite{Kadhe2017} requires a field size $q\geq 2K-M$ ($>K$).\footnote{In the MDS Code scheme of~\cite{Kadhe2017}, the user queries the server to send $K-M$ parity symbols of a systematic $(2K-M,K)$ MDS code; whereas the coded symbols sent by the server in the GRS Code protocol may not correspond to the parity symbols of a systematic $(2K-M,K)$ MDS code. This is because for a systematic MDS code, the parity part of a generator matrix needs to be super-regular~\cite[Proposition 11.4]{Roth:06}. This is while a Vandermonde matrix over a finite field is not guaranteed to be super-regular.}

\begin{lemma}\label{lem:2}
The GRS Code protocol is a PIR-SI protocol, and achieves the rate $D/(K-M)$. 
\end{lemma}

\begin{proof}
Recall that $H(X_1)=\dots=H(X_K) = L$. Since $A_1,\dots,A_{K-M}$ are linearly independent combinations of $X_1,\dots,X_K$, which are themselves independently and uniformly distributed over $\mathbb{F}_{q^{m}}$, then $A_1,\dots,A_{K-M}$ are independently and uniformly distributed over $\mathbb{F}_{q^{m}}$. (The linear independence follows from our choice of the coefficients of $X_j$'s in the linear combinations $A_i$'s, and the full-rank property of the Vandermonde matrix associated with these coefficients.) Then, $H(A_1)=\dots=H(A_{K-M})=L$, and $H(A^{[W,S]}) = H(A_1,\dots,A_{K-M})=\sum_{i=1}^{K-M} H(A_i)=(K-M)L$ for any $W\subset [K]$, $|W|=D$ and any $S\subseteq [K]\setminus W$, $|S|=M$. Thus, the rate of the GRS Code protocol is equal to $DL/H(A^{[\boldsymbol{W},\boldsymbol{S}]}) = DL/H(A^{[W,S]}) = D/(K-M)$, noting that $H(A^{[\boldsymbol{W},\boldsymbol{S}]})=H(A^{[W,S]})$ by the uniformity of joint distribution of $\boldsymbol{W}$ and $\boldsymbol{S}$. 

Next, we prove that the GRS Code protocol is a PIR-SI protocol. Once the user subtracts off the contribution of side information $X_S$ from the $K-M$ linear equations in the answer, the coefficient matrix associated with the resulting system of linear equations is a $(K-M)\times (K-M)$ sub-matrix of a $(K-M)\times K$ Vandermonde matrix. Since this sub-matrix is invertible, the user can uniquely solve the underlying system of linear equations, and recover all $K-M$ messages in $X_{[K]\setminus S}$, including the demand $X_W$. Thus the recoverability condition is satisfied. 

The privacy condition is also satisfied, simply because given $D$ and $M$, the user sends exactly the same query for any demand index set $W$ of size $D$, and any side information index set $S$ of size $M$. Thus, the server does not gain any knowledge about the realization of $W$ (and $S$).
\end{proof}


\section{PIR-SI Problem: The case of $D\leq M$}\label{sec:PIRSI-II}

\subsection{Proof of Achievability for Theorem~\ref{thm:2}}

In this section, we propose a PIR-SI protocol for arbitrary $K$, $M$, and $D\leq M$.


Define $\alpha\triangleq \lfloor M/D \rfloor$, $\beta \triangleq D+\alpha$, $\gamma\triangleq\lfloor K/\beta\rfloor$, and $\rho = K-\beta\gamma$. (Note that $0\leq \rho<\beta$.) Also, define ${\sigma\triangleq \max\{\rho-D,0\}}$. Assume $q\geq \beta$, and let $\omega_1,\dots,\omega_{\beta}$ be $\beta$ distinct elements from $\mathbb{F}_q$. 

\textit{\bf Generalized Partition and Code (GPC) Protocol:} This protocol consists of four steps as follows: 

\textbf{\it Step 1:} First, the user constructs a set $Q_{0}$ of size $\rho$ from the indices in $[K]$, and $\gamma$ disjoint sets $Q_1,\dots,Q_{\gamma}$ (also disjoint from $Q_0$), each of size $\beta$, from the indices in $[K]$. The user randomly chooses $D$ positions among all $K$ positions available in $Q_0,\dots,Q_{\gamma}$, and randomly places the $D$ demand indices in $W$ into these positions. If $Q_0$ (or respectively, $Q_i$ for $i\in [\gamma]$) contains a demand index, the user randomly selects $\sigma$ (or respectively, $\alpha$) indices from $S$ that were not previously selected and positioned, and places them into $Q_0$ (or respectively, $Q_i$). Then, the user randomly places the rest of the indices in $[K]$, that are yet to be placed, into the remaining positions in $Q_0,\dots,Q_{\gamma}$. 

Next, the user creates a collection $Q'$ of $\rho-\sigma$ sequences $Q'_{1},\dots,Q'_{\rho-\sigma}$, each of length $\rho$, such that $Q'_{i}  = \{\omega^{i-1}_1,\dots,\omega^{i-1}_{\rho}\}$ for $i\in [\rho-\sigma]$, and a collection $Q''$ of $D$ sequences $Q''_{1},\dots,Q''_{D}$, each of length $\beta$, such that $Q''_{i}  = \{\omega^{i-1}_1,\dots,\omega^{i-1}_{\beta}\}$ for $i\in [D]$.

\textbf{\it Step 2:} The user constructs $Q^{*}_0 = (Q_0,Q')$ and $Q^{*}_i = (Q_i,Q'')$ for $i\in [\gamma]$, and sends to the server the query $Q^{[W,S]} = \{Q^{*}_{0},\dots,Q^{*}_{\gamma}\}$.

\textbf{\it Step 3:} By using $Q^{*}_0 = (Q_0,Q')$ and $Q^{*}_i = (Q_i,Q'')$, the server computes $A_0 = \{A^1_0,\dots,A^{\rho-\sigma}_0\}$ by $A^j_0 = \sum_{l=1}^{\rho} \omega_l^{j-1} X_{i_l}$ for $j\in [\rho-\sigma]$ where $Q_0 = \{i_1,\dots,i_{\rho}\}$, and computes $A_i = \{A^1_i,\dots,A^D_i\}$ for $i\in [\gamma]$ by $A^{j}_{i} = \sum_{l=1}^{\beta} \omega^{j-1}_{l} X_{i_j}$ for $j\in [D]$ where $Q_i = \{i_1,\dots,i_{\beta}\}$, and sends to the user the answer $A^{[W,S]}=\{A_{0},\dots,A_{\gamma}\}$.

\textbf{\it Step 4:} Upon receiving the answer from the server, the user retrieves $X_j$ for each $j\in W$ by subtracting off the contribution of side information $\{X_i\}_{i\in S}$ from the $D$ equations in $A_{0}$ if $j\in Q_0$, or from the $D$ equations in $A_i$ if $j\in Q_i$, and solving the resulting system of $D$ linear equations with $D$ unknowns. 


The following demonstrates an example where the GPC protocol achieves a higher rate than the GRS Code protocol. 

\begin{example}
Consider a scenario where the server has ${K=10}$ messages $X_1,\dots,X_{10}\in \mathbb{F}_{5}$, and the user's demand and side information index sets are respectively $W=\{3,4\}$ and $S=\{5,8\}$ (i.e., $D=2$ and $M=2$).

The GPC protocol's parameters for this example are: ${\alpha = 1}$, $\beta = 3$, $\gamma = 3$, $\rho = 1$, $\sigma = 0$, and $\{\omega_1,\omega_2,\omega_3\}=\{0,1,2\}$. First, the user creates four sets $Q_0$, $Q_1$, $Q_2$, and $Q_3$, where $Q_0=\{ - \}$ has one position (slot) to be filled, and ${Q_1=\{-,-,-\}}$, ${Q_2=\{-,-,-\}}$, and ${Q_3=\{-,-,-\}}$ have three slots each. The user then randomly chooses two slots (out of the $10$ slots in total) to place the demand indices $3$ and $4$. Say that the user places $4$ in one of the slots in $Q_1$, and places $3$ into one of the slots in $Q_2$, i.e., $Q_0=\{-\}$, ${Q_1=\{4,-,-\}}$, ${Q_2 =\{3,-,-\}}$, and ${Q_3=\{-,-,-\}}$. (The order of the slots within the same set $Q_i$ is irrelevant.) The user then places one randomly chosen side information index in $Q_1$ and the other in $Q_2$. Say that the user randomly chooses the side information index $8$ to place in $Q_1$, i.e., ${Q_0=\{ - \}}$, ${Q_1=\{4,8,-\}}$, ${Q_2=\{3,5,-\}}$, and ${Q_3=\{-,-,-\}}$. Then the user randomly places the rest of the indices into the empty slots in these three sets; say ${Q_0= \{2\}}$, ${Q_1=\{4,6,8\}}$, ${Q_2=\{3,5,7\}}$, and ${Q_3=\{1,9,10\}}$. 

Next, the user forms the collection $Q'$, with the sequence $Q'_1$, which in this example is a sequence of length one, with the element $\omega_1=0$: $Q'_1=\{1\}$; and forms the collection $Q''$ of two sequences $Q''_1$ and $Q''_2$ using the three elements ${\omega_1=0}$, $\omega_2=1$, and $\omega_3=2$: ${Q''_1=\{1,1,1\}}$ and ${Q''_2=\{0,1,2\}}$. The user then sends to the server 
\begin{align*}
(Q_0,Q') &=(\{2\},\{1\}),\\
(Q_1,Q'') &=(\{4,6,8\},\{\{1,1,1\},\{0,1,2\}\}),\\
(Q_2,Q'') &=(\{3,5,7\},\{\{1,1,1\},\{0,1,2\}\}), \\
(Q_3,Q'') &=(\{1,9,10\},\{\{1,1,1\},\{0,1,2\}\}),
\end{align*} and the server sends the user back 
\begin{align*}
A_0 &=\{X_2\},\\
A_1 &=\{{X_4+X_6+X_8},{X_6+2X_8}\},\\
A_2 &=\{{X_3+X_5+X_7},{X_5+2X_7}\}.\\
A_3 &= \{{X_1+X_9+X_{10}},{X_9+2X_{10}} \}.
\end{align*} The user then solves for $X_3$ (and $X_7$) by subtracting off the contribution of $X_5$ from the equations in $A_2$ and solving the resulting equations, and solves for $X_4$ (and $X_6$) by subtracting off the contribution of $X_8$ from the equations in $A_1$, and solving the resulting equations.

The rate of the GPC protocol for this example is $2/7$, whereas the GRS Code protocol achieves a lower rate ${D/(K-M)=2/8}$. 
\end{example}

The next example illustrates a scenario in which the GRS Code protocol achieves a higher rate than the GPC protocol.

\begin{example}
Consider a scenario where the server has ${K=5}$ messages $X_1,X_2,\dots,X_5\in \mathbb{F}_{5}$, and the user's demand and side information index sets are respectively $W=\{2,5\}$ and $S=\{1,3\}$ (i.e. $D=2$ $M=2$). 


In the GPC protocol, the parameters defined for this example are: $\alpha = 1$, $\beta = 3$, $\gamma=1$, $\rho=2$, $\sigma=0$, and $\{\omega_1, \omega_2, \omega_3 \} = \{0,1,2\}$. The user creates two sets $Q_0$ and $Q_1$, where $Q_0 = \{-,-\}$ and $Q_1=\{-,-,-\}$. The user randomly chooses two slots to place $1$ and $3$, say after choosing slots, $Q_0=\{5,-\}$ and $Q_1=\{2,-,-\}$. The user then places one element from $S$ into $Q_1$, say $Q_1=\{1,2,-\}$. Then the user fills the rest of the slots with unplaced elements; say $Q_0=\{3,5\}$ and $Q_1=\{1,2,4\}$. 

Then the user forms the collection $Q'$ by using the elements $\omega_1=0$ and $\omega_2=1$, and the collection of sequences $Q''$ by using the elements $\omega_1=0$, $\omega_2=1$, and $\omega_3=2$. The following are the collections:
\begin{align*}
Q' &= \{\{1,1\},\{0,1\}\}, \\
Q'' &= \{\{1,1,1\}, \{0,1,2\}\}.
\end{align*} The user then sends the server $Q=\{(Q_0,Q'),(Q_1,Q'')\}$, and the server sends the user back the following equations:
\begin{align*}
A_0 &= \{X_3+X_5,X_5\},\\
A_1 &= \{X_1+X_2+X_4,X_2+2X_4\}.
\end{align*} The user then recovers $X_5$ (and $X_3$) from $A_0$, and recovers $X_2$ (and $X_4$) by subtracting $X_1$ from the two equations in $A_1$, and then solving for $X_2$ (and $X_4$). 

For this example, the rate of the GPC protocol is $1/2$, and the rate of the GRS Code protocol is ${D/(K-M)=2/3}$. 
\end{example}

\begin{lemma}\label{lem:3}
The Generalized Partition and Code (GPC) protocol is a PIR-SI protocol, and achieves the rate \[D\left(K-\left\lfloor\frac{M}{D}\right\rfloor\left\lfloor\frac{K}{D+\lfloor M/D\rfloor}\right\rfloor\right)^{-1}\] if $\frac{K-D}{D+\lfloor M/D\rfloor}\leq \left\lfloor\frac{K}{D+\lfloor M/D\rfloor}\right\rfloor$, and \[\left\lceil\frac{K}{D+\lfloor M/D\rfloor}\right\rceil^{-1}\] if $\frac{K-D}{D+\lfloor M/D\rfloor}> \left\lfloor\frac{K}{D+\lfloor M/D\rfloor}\right\rfloor$.  	
\end{lemma}

\begin{proof} 
Recall that $H(X_1) = \dots = H(X_K) = L$. If $\frac{K-D}{D+\lfloor M/D\rfloor}\leq \lfloor\frac{K}{D+\lfloor M/D \rfloor}\rfloor$, then $0\leq \rho\leq D$. Thus, $\sigma = 0$, and accordingly, ${\rho-\sigma = \rho}$. In this case, $H(A_0) = \rho L$ and $H(A_i) = DL$ for $i\in [\gamma]$, and consequently, $H(A^{[W,S]})=H(A_0,\dots,A_{\gamma})=\sum_{i=0}^{\gamma} H(A_i)=(\rho+\gamma D)L$ (by the independence of $A_0,\dots,A_{\gamma}$), for any $W\subset [K]$, $|W|=D$ and any $S\subseteq [K]\setminus W$, $|S|=M$. Thus, in this case, the rate, i.e., $DL/H(A^{[\boldsymbol{W},\boldsymbol{S}]}) = DL/H(A^{[W,S]})$, is equal to \[\frac{D}{(\rho+\gamma D)} = {D\left(K-\left\lfloor\frac{M}{D}\right\rfloor \left\lfloor\frac{K}{D+\left\lfloor M/D\right\rfloor}\right\rfloor\right)^{-1}}.\] If $\frac{K-D}{D+\lfloor M/D\rfloor}> \lfloor\frac{K}{D+\lfloor M/D \rfloor}\rfloor$, then $\rho> D$. Thus, ${\sigma = \rho-D}$, and accordingly, $\rho-\sigma= D$. In this case, $H(A_0) = H(A_i) = DL$ for $i\in [\gamma]$, and as a consequence, ${H(A^{[W,S]})=(\gamma+1)DL}$. The rate in this case is equal to \[\frac{D}{(D+\gamma D)} = \left\lceil\frac{K}{D+\lfloor M/D\rfloor}\right\rceil^{-1}.\] 

Next, we prove that the GPC protocol is a PIR-SI protocol. By the structure of the GPC protocol, it is easy to see that the recoverability condition is satisfied. Then, it remains to prove that the GPC protocol satisfies the privacy condition. The details of the proof are as follows. 

 



Consider an arbitrary ${Q = \{Q_0,\dots,Q_{\gamma}\}}$ (see Step~1 of the GPC protocol for definition of $Q_i$'s). Consider arbitrary $W\subset [K]$, $|W|=D$, and ${S\subseteq [K]\setminus W}$, ${|S|=M}$ such that $W$ and $S$ comply with $Q$, i.e., given that $W$ and $S$ are respectively the user's demand and side information index sets, the GPC protocol could potentially construct $Q$. 

Define $n_0,\dots,n_{\gamma}$ as the number of indices from $W$ that are placed into the sets $Q_0,\dots,Q_{\gamma}$, respectively. Note that ${0\leq n_0\leq \rho-\sigma}$ and $0\leq n_i\leq D$ for $i\in [\gamma]$. Note also that ${\sum_{i=0}^{\gamma} n_i = D}$. Let $\lambda$ be the number of sets $\{Q_i\}_{i\in [\gamma]}$ such that ${n_i\neq 0}$, i.e., $\lambda=\sum_{i=1}^{\gamma} \mathds{1}_{\{n_i\neq 0\}}$. Assume, without loss of generality, that ${n_1,\dots,n_{\lambda}\neq 0}$ and ${n_{\lambda+1},\dots,n_{\gamma}=0}$. Note that no assumption is made on $n_0$. Let ${m_j\triangleq 1}$ for all ${j\in [D]\setminus \{n_1,\dots,n_{\lambda}\}}$, and let ${m_j \triangleq \sum_{i=0}^{\lambda} \mathds{1}_{\{n_i=j\}}}$ for all $j\in \{n_1,\dots,n_{\lambda}\}$.

In order to prove that the GPC protocol satisfies the privacy condition, we need to show that 
\begin{align*}
\mathbb{P}(\boldsymbol{W} = W| \boldsymbol{Q} = Q,\boldsymbol{X}=X) = \frac{1}{\binom{K}{D}}.	
\end{align*}

Since $Q$ does not depend on the contents of the messages in $X$ (by the structure of the GPC protocol), then ${\mathbb{P}(\boldsymbol{W}=W|\boldsymbol{Q}=Q,\boldsymbol{X}=X)=\mathbb{P}(\boldsymbol{W}=W|\boldsymbol{Q}=Q)}$. Thus, it suffices to show that
\begin{align}\label{eq:1}
\mathbb{P}(\boldsymbol{W} = W| \boldsymbol{Q} = Q) = \frac{1}{\binom{K}{D}}.	
\end{align}  

Since $N= \{n_1,\dots,n_{\gamma}\}$ is uniquely determined by $Q$ and $W$, we can write
\begin{align}\nonumber
\mathbb{P}(\boldsymbol{W}=W|\boldsymbol{Q} = Q) &= \mathbb{P}(\boldsymbol{W}=W,\boldsymbol{N} = N|\boldsymbol{Q}=Q) \\ \nonumber 
& = \mathbb{P}(\boldsymbol{N}=N|\boldsymbol{Q}=Q)\\\label{eq:2} & \quad \times \mathbb{P}(\boldsymbol{W}=W|\boldsymbol{N} = N,\boldsymbol{Q} = Q)	.
\end{align} Thus, we need to compute $\mathbb{P}(\boldsymbol{N}=N|\boldsymbol{Q}=Q)$ and ${\mathbb{P}(\boldsymbol{W}=W|\boldsymbol{N} = N, \boldsymbol{Q}=Q)}$. 

First, we show that 
\begin{align}\nonumber
\mathbb{P}(\boldsymbol{N}=N|\boldsymbol{Q}=Q) & = \mathbb{P}(\boldsymbol{N} = N)\\ \label{eq:3}
& = \frac{\binom{\gamma}{\lambda}\binom{\lambda}{m_1,\dots,m_D}\binom{\rho}{n_0}\prod_{i=1}^{\lambda}\binom{\beta}{n_i}}{\binom{K}{D}}.
\end{align} Note that 
\begin{align*}
& \mathbb{P}(\boldsymbol{Q}=Q|\boldsymbol{N}=N) = \\  & \quad \sum \mathbb{P}(\boldsymbol{Q} = Q,\boldsymbol{W}=W,\boldsymbol{S}=S|\boldsymbol{N}=N)
\end{align*} where the sum is over all $W,S$. Since the events ${(\boldsymbol{W}=W,\boldsymbol{S}=S)}$ and $\boldsymbol{N}=N$ are independent, ${\mathbb{P}(\boldsymbol{Q}=Q|\boldsymbol{N}=N)}$ can be written as
\begin{align*}
& \sum \mathbb{P}(\boldsymbol{W}=W,\boldsymbol{S}=S)\\ &\quad \hspace{0.125cm}\times \mathbb{P}(\boldsymbol{Q}=Q|\boldsymbol{W}=W,\boldsymbol{S}=S,\boldsymbol{N}=N)
\end{align*} where the sum is over all $W,S$ that comply with $Q$ given $N$. By assumption,
\begin{align}\label{eq:4}
\mathbb{P}(\boldsymbol{W}=W,\boldsymbol{S}=S) =\frac{1}{\binom{K}{D}\binom{K-D}{M}}.	
\end{align} 

Next, we compute $\mathbb{P}(\boldsymbol{Q}=Q|\boldsymbol{W}=W,\boldsymbol{S}=S,\boldsymbol{N}=N)$. Since $N = \{n_0,\dots,n_{\gamma}\}$, the number of ways that the GPC protocol could position the $D$ indices from $W$ into $Q_0,\dots,Q_{\lambda}$ (for a given $N$) is $\binom{D}{n_0,\dots,n_{\lambda}}$. Since all the $m_j$ parts of the same size $j\in [D]$ can be arbitrarily permuted without any change in the structure of the query, the number of distinct ways for placing the $D$ indices from $W$ in $Q_0,\dots,Q_{\lambda}$ (for a given $N$) is \[\frac{1}{\prod_{i=1}^{D} m_i!}\binom{D}{n_0,\dots,n_{\lambda}} = \frac{D!}{\prod_{i=1}^{D} m_i! \prod_{i=0}^{\lambda} n_i!}.\] For placing indices from $S$ next to the previously-placed indices from $W$, the GPC protocol randomly selects $\sigma+\lambda\alpha$ indices from the $M$ indices in $S$, particularly, $\sigma$ indices to be placed in $Q_0$, and $\alpha$ indices to be placed in $Q_i$ for each $i\in [\lambda]$. This selection of indices can be done in $\binom{M}{\sigma+\lambda\alpha}$ ways. (Note that $\sigma+\lambda\alpha\leq M$.) The GPC protocol randomly partitions the $\sigma+\lambda\alpha$ indices selected from $S$ into $\lambda+1$ groups as follows: a group of size $\sigma$, and $\lambda$ groups each of size $\alpha$. The number of ways for such a partitioning is \[\binom{\sigma+\lambda\alpha}{\sigma,\alpha,\dots,\alpha} = \frac{(\sigma+\lambda\alpha)!}{\sigma! (\alpha!)^{\lambda}}.\] The remaining $K-D-\sigma-\lambda\alpha$ indices in $[K]$, yet to be placed by the GPC protocol, can be randomly partitioned into $\gamma+1$ groups as follows: a group of size $\rho-\sigma-n_0$, $\lambda$ groups of sizes $\{\beta-\alpha-n_{i}\}$ for all $i\in [\lambda]$, and $\gamma-\lambda$ groups each of size $\beta$. Since the $\gamma-\lambda$ groups of size $\beta$ can be arbitrarily permuted without changing the structure of the query, the number of distinct ways of such a partitioning is 
\begin{align*}
& \frac{1}{(\gamma-\lambda)!} \\ &\hspace{0.125cm} \times \binom{K-D-\sigma-\lambda\alpha}{\rho-\sigma-n_0,\beta-\alpha-n_{1},\dots,\beta-\alpha-n_{\lambda},\beta,\dots,\beta} \\ & = \frac{(K-D-\sigma-\lambda\alpha)!}{(\gamma-\lambda)!(\rho-\sigma-n_0)!\prod_{i=1}^{\lambda}(\beta-\alpha-n_{i})!(\beta!)^{\gamma-\lambda}}.	
\end{align*} Putting these arguments together, we can write
 \begin{align}\nonumber
 & \mathbb{P}(\boldsymbol{Q} = Q|\boldsymbol{W} = W,\boldsymbol{S} = S, \boldsymbol{N} = N) = \\ \nonumber
 & \quad \frac{\prod_{i=1}^{D} m_i! \prod_{i=0}^{\lambda} n_i!(M-\sigma-\lambda\alpha)!\sigma! (\alpha!)^{\lambda}}{D!M!}	\\ \label{eq:5} &\quad \times \frac{(\gamma-\lambda)!(\rho-\sigma-n_0)!\prod_{i=1}^{\lambda}(\beta-\alpha-n_{i})!(\beta!)^{\gamma-\lambda}}{(K-D-\sigma-\lambda\alpha)!}.
 \end{align} 

Let $T_1$ be the number of $(W,S)$ that comply with $Q$ given $N$. To compute $T_1$, we proceed as follows. There are $\binom{\rho}{n_0}$ ways to choose $n_0$ indices from $Q_0$ to be demand indices. There are $\binom{\gamma}{\lambda}$ ways to choose $\lambda$ sets $Q_{1},\dots,Q_{\lambda}$ from $Q_1,\dots,Q_{\gamma}$, and 
\begin{align*}
& \binom{\lambda}{m_1,\dots,m_D}\binom{\beta}{n_{1}}\cdots \binom{\beta}{n_{\lambda}} \\ & \quad = \frac{\lambda!(\beta!)^{\lambda}}{\prod_{i=1}^{D} m_i! \prod_{i=1}^{\lambda} n_{i}! \prod_{i=1}^{\lambda}(\beta-n_{i})!}	
\end{align*} ways to choose $D$ indices from $Q_{1},\dots,Q_{\lambda}$, particularly $n_{i}$ indices from $Q_{i}$ for each $i\in [\lambda]$, to be demand indices. There are $\binom{\rho-n_0}{\sigma}$ ways to choose $\sigma$ indices from $Q_0$ to be side information indices, and 
\begin{align*}
\binom{\beta-n_{1}}{\alpha}\cdots \binom{\beta-n_{\lambda}}{\alpha} = \frac{\prod_{i=1}^{\lambda}(\beta-n_i)!}{(\alpha!)^{\lambda}\prod_{i=1}^{\lambda}(\beta-\alpha-n_i)!}	
\end{align*} ways to choose $\lambda\alpha$ indices from $Q_{1},\dots,Q_{\lambda}$, particularly, $\alpha$ indices from $Q_{i}$ for each $i\in [\lambda]$, to be side information indices. The total number of ways to choose the remaining ${M-\sigma-\lambda\alpha}$ indices to be side information indices is \[\binom{K-D-\sigma-\lambda\alpha}{M-\sigma-\lambda\alpha} = \frac{(K-D-\sigma-\lambda\alpha)!}{(M-\sigma-\lambda\alpha)!(K-D-M)!}.\] By combining the above arguments, 
\begin{align*}
T_1 & = \binom{\rho}{n_0}\binom{\gamma}{\lambda}\binom{\lambda}{m_1,\dots,m_D}\prod_{i=1}^{\lambda}\binom{\beta}{n_i} \\ &\quad \times \binom{\rho-n_0}{\sigma}\prod_{i=1}^{\lambda}\binom{\beta-n_i}{\alpha}\binom{K-D-\sigma-\lambda\alpha}{M-\sigma-\lambda\alpha}\\ & = \frac{\rho!\gamma!(\beta!)^{\lambda}}{(\gamma-\lambda)!\prod_{i=1}^{D} m_i! \prod_{i=0}^{\lambda}n_i! \prod_{i=1}^{\lambda}(\beta-\alpha-n_i)!\sigma!} \\ 
& \quad \times	\frac{(K-D-\sigma-\lambda\alpha)!}{(\rho-\sigma-n_0)!(\alpha!)^{\lambda}(M-\sigma-\lambda\alpha)!(K-D-M)!}
\end{align*}

Since ${\mathbb{P}(\boldsymbol{Q} = Q|\boldsymbol{W} = W,\boldsymbol{S} = S, \boldsymbol{N} = N)}$ and ${\mathbb{P}(\boldsymbol{W}=W,\boldsymbol{S}=S)}$ do not depend on indices in $W$ and $S$ (by~\eqref{eq:4} and~\eqref{eq:5}), we can write
\begin{align}\nonumber
& \mathbb{P}(\boldsymbol{Q}=Q|\boldsymbol{N}=N) \\ \nonumber & \quad= T_1\times \mathbb{P}(\boldsymbol{W}=W,\boldsymbol{S}=S) \\ \nonumber & \quad \quad\times \mathbb{P}(\boldsymbol{Q} = Q|\boldsymbol{W} = W,\boldsymbol{S} = S, \boldsymbol{N} = N)\\ \label{eq:6}
& \quad = \frac{\gamma!\rho!(\beta!)^{\gamma}}{K!}.	
\end{align} 

From the perspective of the server, all partitions $Q$ are equiprobable (by construction), and the total number of distinct partitions is \[\frac{1}{\gamma!}\binom{K}{\rho,\beta,\dots,\beta} = \frac{K!}{\gamma!\rho!(\beta!)^{\gamma}}.\] Thus, 
\begin{equation}\label{eq:7}
\mathbb{P}(\boldsymbol{Q} = Q) = \frac{\gamma ! \rho!(\beta !)^{\gamma}}{K!}.	
\end{equation} 

By~\eqref{eq:6} and~\eqref{eq:7}, the events $\boldsymbol{Q}=Q$ and $\boldsymbol{N}=N$ are independent. Thus, $\mathbb{P}(\boldsymbol{N}=N|\boldsymbol{Q}=Q) = \mathbb{P}(\boldsymbol{N}=N)$. Since $N$ can be chosen in \[\binom{\gamma}{\lambda}\binom{\lambda}{m_1,\dots,m_D}\binom{\rho}{n_0}\prod_{i=1}^{\lambda}\binom{\beta}{n_i}\] ways (as argued earlier), and there are $\binom{K}{D}$ ways in total to place the $D$ indices from $W$ into $Q_0,\dots,Q_{\gamma}$, we get
\begin{equation*}\label{eq:8}
\mathbb{P}(\boldsymbol{N}=N) = \frac{\binom{\gamma}{\lambda}\binom{\lambda}{m_1,\dots,m_D}\binom{\rho}{n_0}\prod_{i=1}^{\lambda}\binom{\beta}{n_i}}{\binom{K}{D}}.	
\end{equation*} This completes the proof of~\eqref{eq:3}.

Next, we shall show that 
\begin{align}\nonumber
& \mathbb{P}(\boldsymbol{W}=W|\boldsymbol{N}=N,\boldsymbol{Q}=Q) \\ \label{eq:8} & \quad = \frac{1}{\binom{\gamma}{\lambda}\binom{\lambda}{m_1,\dots,m_D}\binom{\rho}{n_0}\prod_{i=1}^{\lambda}\binom{\beta}{n_i}}.	
\end{align} Note that 
\begin{align}\nonumber
& \mathbb{P}(\boldsymbol{W}=W|\boldsymbol{N}=N,\boldsymbol{Q}=Q) \\ \label{eq:10} & \quad = \sum \mathbb{P}(\boldsymbol{W}=W,\boldsymbol{S}=S|\boldsymbol{N}=N,\boldsymbol{Q}=Q)
\end{align} where the sum is over all $S$ such that $W,S$ comply with $Q$ given $N$. Let $T_2$ be the number of such $S$. By using similar arguments as for the case of $T_1$, 
\begin{align*}
T_2 & = \binom{\rho-n_0}{\sigma}\prod_{i=1}^{\lambda}\binom{\beta-n_i}{\alpha}\binom{K-D-\sigma-\lambda\alpha}{M-\sigma-\lambda\alpha} \\ & = 	\frac{(\rho-n_0)!(\rho-\sigma-n_0)!\prod_{i=1}^{\lambda}(\beta-n_i)!}{\sigma!(\alpha!)^{\lambda}\prod_{i=1}^{\lambda}(\beta-\alpha-n_i)!} \\ &\quad \times \frac{(K-D-\sigma-\lambda\alpha)!}{(M-\sigma-\lambda\alpha)!(K-D-M)!}
\end{align*} 

By applying the Bayes' rule and the chain rule of conditional probability, we can write
\begin{align}\nonumber 
&\mathbb{P}(\boldsymbol{W}=W,\boldsymbol{S}=S|\boldsymbol{N}=N,\boldsymbol{Q}=Q) \\ \nonumber & \quad = \mathbb{P}(\boldsymbol{W}=W,\boldsymbol{S}=S) \\ \nonumber &\quad \quad \times \frac{\mathbb{P}(\boldsymbol{N}=N,\boldsymbol{Q}=Q|\boldsymbol{W}=W,\boldsymbol{S}=S)}{\mathbb{P}(\boldsymbol{N}=N,\boldsymbol{Q}=Q)}\\ \nonumber & \quad = \mathbb{P}(\boldsymbol{W}=W,\boldsymbol{S}=S)\mathbb{P}(\boldsymbol{N}=N|\boldsymbol{W}=W,\boldsymbol{S}=S) \\ \label{eq:11} &\quad \quad \times \frac{\mathbb{P}(\boldsymbol{Q}=Q|\boldsymbol{W}=W,\boldsymbol{S}=S,\boldsymbol{N} = N)}{\mathbb{P}(\boldsymbol{N}=N)\mathbb{P}(\boldsymbol{Q}=Q)}. 
\end{align} By the independence of the events $(\boldsymbol{W}=W,\boldsymbol{S}=S)$ and ${\boldsymbol{N}=N}$, $\mathbb{P}(\boldsymbol{N}=N|\boldsymbol{W}=W,\boldsymbol{S}=S) = \mathbb{P}(\boldsymbol{N}=N)$. Rewriting~\eqref{eq:11}, 
\begin{align}\nonumber
&\mathbb{P}(\boldsymbol{W}=W,\boldsymbol{S}=S|\boldsymbol{N}=N,\boldsymbol{Q}=Q)	 \\ \nonumber &\quad =\mathbb{P}(\boldsymbol{W}=W,\boldsymbol{S}=S) \\ \label{eq:12}
& \quad \quad \times \frac{\mathbb{P}(\boldsymbol{Q}=Q|\boldsymbol{W}=W,\boldsymbol{S}=S,\boldsymbol{N} = N)}{\mathbb{P}(\boldsymbol{Q}=Q)}.
\end{align} By using~\eqref{eq:4},~\eqref{eq:5}, and~\eqref{eq:7}, and rewriting~\eqref{eq:12}, 
\begin{align}\nonumber
& \mathbb{P}(\boldsymbol{W}=W,\boldsymbol{S}=S|\boldsymbol{N}=N,\boldsymbol{Q}=Q)	 \\ \nonumber & \quad = \frac{\prod_{i=1}^{D} m_i ! \prod_{i=0}^{\lambda} n_i ! (M-\sigma-\lambda\alpha)!\sigma!(\alpha!)^{\lambda}(\gamma-\lambda)!}{\gamma!\rho!(\beta!)^{\lambda}}\\ \label{eq:13} &\quad \quad \times\frac{(K-D-M)!(\rho-\sigma-n_0)!\prod_{i=1}^{\lambda}(\beta-\alpha-n_i)!}{(K-D-\sigma-\lambda\alpha)!}
\end{align} Since $\mathbb{P}(\boldsymbol{W}=W,\boldsymbol{S}=S|\boldsymbol{N}=N,\boldsymbol{Q}=Q)$ does not depend on the indices of $W$ and $S$ (by~\eqref{eq:13}), we can rewrite~\eqref{eq:10} as
\begin{align}\nonumber
& \mathbb{P}(\boldsymbol{W}=W|\boldsymbol{N}=N,\boldsymbol{Q}=Q) \\ \label{eq:14} & \quad = T_2\times \mathbb{P}(\boldsymbol{W}=W,\boldsymbol{S}=S|\boldsymbol{N}=N,\boldsymbol{Q}=Q) 	
\end{align} By combining~\eqref{eq:13} and~\eqref{eq:14}, 
\begin{align}\nonumber
& \mathbb{P}(\boldsymbol{W}=W|\boldsymbol{N}=N,\boldsymbol{Q}=Q) \\ \nonumber & \hspace{0.125cm} = \frac{(\gamma-\lambda)!\prod_{i=1}^{D} m_i !(\rho-n_0)! \prod_{i=0}^{\lambda} n_i!\prod_{i=1}^{\lambda}(\beta-n_i)!}{\gamma! \rho! (\beta!)^{\lambda}} 
\\ \nonumber & \hspace{0.125cm} = \frac{1}{\binom{\gamma}{\lambda}\binom{\lambda}{m_1,\dots,m_D}\binom{\rho}{n_0}\prod_{i=1}^{\lambda}\binom{\beta}{n_i}}.	
\end{align} This completes the proof of~\eqref{eq:8}. 

Putting~\eqref{eq:2},~\eqref{eq:3}, and~\eqref{eq:8} together, we get
\begin{align*}
\mathbb{P}(\boldsymbol{W}=W|\boldsymbol{Q} = Q) = \frac{1}{\binom{K}{D}}.
\end{align*} This proves that the GPC protocol satisfies the privacy condition. 
\end{proof}


\bibliographystyle{IEEEtran}
\bibliography{PIR_salim,pir_bib,coding1,coding2}

\appendix[Tightness of the Result of Theorem~\ref{thm:2} for $D=1$]


For the case of $D=1$, the result of Theorem~\ref{thm:2} shows that the capacity is lower bounded by ${(K-M\lfloor \frac{K}{M+1}\rfloor)^{-1}}$ if ${\frac{K-1}{M+1}\leq \lfloor \frac{K}{M+1}\rfloor}$, and is lower bounded by ${\lceil \frac{K}{M+1}\rceil^{-1}}$ if ${\frac{K-1}{M+1}> \lfloor \frac{K}{M+1}\rfloor}$. In order to show that these lower bounds are equal, it suffices to show that \[K-M\left\lfloor \frac{K}{M+1}\right\rfloor = \left\lceil\frac{K}{M+1}\right\rceil\] when $\frac{K-1}{M+1}\leq \lfloor\frac{K}{M+1}\rfloor$. 

There are two cases: (i) $M+1$ divides $K$, and (ii) $M+1$ does not divide $K$. In the case (i), \[K-M\left\lfloor\frac{K}{M+1}\right\rfloor = K-\frac{MK}{M+1} = \frac{K}{M+1} = \left\lceil \frac{K}{M+1}\right\rceil,\] noting that $\frac{K}{M+1}$ is an integer in this case. This completes the proof for the case (i). 

In the case (ii), from $\frac{K-1}{M+1}\leq \lfloor\frac{K}{M+1}\rfloor$ it follows that 
\begin{align}\nonumber
K-M\left\lfloor \frac{K}{M+1}\right\rfloor & \leq \left\lfloor\frac{K}{M+1}\right\rfloor+1 \\ \label{eq:Last1} & = \left\lceil\frac{K}{M+1}\right\rceil,
\end{align}
noting that $\frac{K}{M+1}$ is not an integer in this case. Since ${\lfloor\frac{K}{M+1}\rfloor< \frac{K}{M+1}}$, then 
\begin{align}\nonumber
K-M\left\lfloor\frac{K}{M+1}\right\rfloor & > K-\frac{MK}{M+1} \\ \nonumber & = \frac{K}{M+1} \\ \label{eq:Last2}& >\left\lfloor\frac{K}{M+1}\right\rfloor.
\end{align} By combining~\eqref{eq:Last1} and~\eqref{eq:Last2}, \[\left\lfloor\frac{K}{M+1}\right\rfloor<K-M\left\lfloor\frac{K}{M+1}\right\rfloor\leq \left\lceil\frac{K}{M+1}\right\rceil,\] and subsequently, \[K-M\left\lfloor\frac{K}{M+1}\right\rfloor = \left\lceil\frac{K}{M+1}\right\rceil.\] The completes the proof for the case (ii). 

\end{document}